\documentclass[a4paper,11pt,reqno]{amsart}

\usepackage[left=2.2cm,right=2.5cm,top=3cm,bottom=3cm]{geometry}

\usepackage{amsmath}
\usepackage{amssymb}
\usepackage{amsthm}
\usepackage{refcount}

\usepackage[dvips]{color}
\usepackage{pgf,tikz}
\usetikzlibrary{shapes,automata,positioning,arrows}
\usepackage{color}
\usepackage{graphicx}
\usepackage{setspace}
\usepackage[version=3]{mhchem}
\usepackage{chemfig} 
\usepackage[sort&compress,comma,square,numbers]{natbib}
\usepackage{booktabs}

\usepackage{caption}
\usepackage{subcaption}

\usepackage{enumerate}
\usepackage{hyperref}

\usepackage{url}

\usepackage{mathabx}
 
\newcommand{\R}{\mathbb{R}}

\DeclareMathOperator*{\im}{im}

\DeclareMathOperator{\diag}{diag}

\renewcommand{\k}{{\kappa}}

\newtheorem{proposition}{Proposition}
\newtheorem{theorem}{Theorem}

\theoremstyle{definition}
\newtheorem{definition}{Definition}

\newtheorem{remark}{Remark}

\begin{document}

\title[Absence of Hopf bifurcations in distributive phosphorylation]{
  On the existence of Hopf bifurcations in the sequential and distributive double phosphorylation cycle
}

\author{Carsten Conradi$^1$, Elisenda Feliu$^2$, Maya Mincheva$^3$ }
\date{\today}

\footnotetext[1]{Life Science Engineering, HTW Berlin, Wilhelminenhoftstr.\ 75, 10459 Berlin. carsten.conradi@htw-berlin.de}
\footnotetext[2]{Department of Mathematical Sciences, University of Copenhagen, Universitetsparken 5, 2100 Copenhagen, Denmark. efeliu@math.ku.dk}
\footnotetext[3]{Department of Mathematical Sciences, Northern Illinois University, 1425 W. Lincoln Hwy.,  DeKalb IL 60115, USA. mmincheva@niu.edu}

\maketitle

\begin{abstract}
  Protein phosphorylation cycles are important mechanisms of the post
  translational modification of a protein and as such an integral part
  of intracellular signaling and control. We consider the sequential
  phosphorylation and dephosphorylation of a protein at two binding
  sites. While it is known that proteins where phosphorylation is
  processive and dephosphorylation is distributive admit 
  oscillations (for some value of the rate constants and total
  concentrations) it is not known whether or not this is the case if
  both phosphorylation and dephosphorylation are distributive. We
  study simplified mass action models of sequential and
  distributive phosphorylation and show that for each of
  those there do not exist rate constants and total concentrations
  where a Hopf bifurcation occurs. To arrive at this result we use
  convex parameters to parametrize the steady state and
  Hurwitz matrices.
\end{abstract}
 
 \section{Introduction}
 
Protein phosphorylation cycles consist of three proteins, a substrate
$S$ and two enzymes $K$ and $F$. One enzyme, the kinase $K$, attaches 
phosphate groups to the substrate and hence phosphorylates the
substrate while the other, the phosphatase $F$, removes phosphate
groups and hence dephosphorylates the substrate. Protein
phosphorylation cycles are important mechanisms of post translational
modification of a protein and as such an integral part of
intracellular signaling and control \cite{conradi-shiu-review}. 
Often phosphorylation and dephosphorylation follow a sequential and
distributive mechanism as depicted in Fig.~\ref{fig:cycles-dd}: in
each encounter of $S$ and either $K$ or $F$ exactly one binding site is
(de)phosphorylated. If either phosphorylation or dephosphorylation follows
a processive mechanism, then at least two binding sites are
(de)phosphorylated in each encounter of $S$ and either $K$ or $F$
(cf. Fig.~\ref{fig:cycles-pd} \& \ref{fig:cycles-dp}).
Here we study the sequential and distributive phosphorylation of a
protein $S$ at two binding sites as depicted in
Fig.~\ref{fig:cycles-dd}.
Such a study of the behavior of an important biochemical module is
of particular interest in the light of studies elucidating the
complex behavior of signaling pathways composed of such modules
\cite{suwanmajo2018exploring}.

Mathematical models of both processive and distributive phosphorylation
have been extensively studied and it is known that they admit complex 
dynamics (see e.g.\ \cite{conradi-shiu-review,G-PNAS,SH09,TG-Nature} and the many references therein). 
The mass action model of the sequential and distributive
phosphorylation cycle depicted in Fig.~\ref{fig:cycles-dd} is
arguably one of the -- mathematically -- best studied and challenging systems of
post translational modification: both multistationarity (the existence
of at least two positive steady states) and bistability (the existence
of two locally stable positive steady states) have been established 
(cf.\ for example \cite{fein-024,conradi-mincheva} for
multistationarity, \cite{rendall-2site} for bistability). In fact it
has been shown that this mass action model admits at most three
positive steady states \cite{Wang:2008dc,FHC14}.

\begin{figure}[!h]
  \centering
    \begin{subfigure}{0.3\textwidth}
    \includegraphics[width=\textwidth]{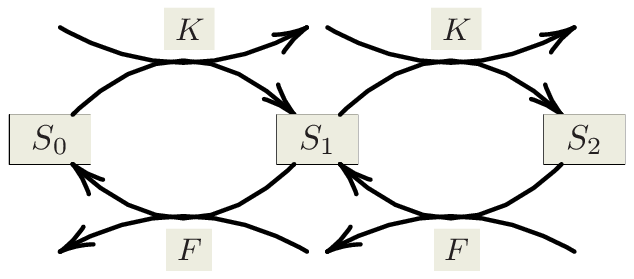}
    \subcaption{
      \label{fig:cycles-dd}
      Distributive phosphorylation and dephosphorylation
    }
  \end{subfigure}
  \hfill
  \begin{subfigure}{0.3\textwidth}
    \includegraphics[width=\textwidth]{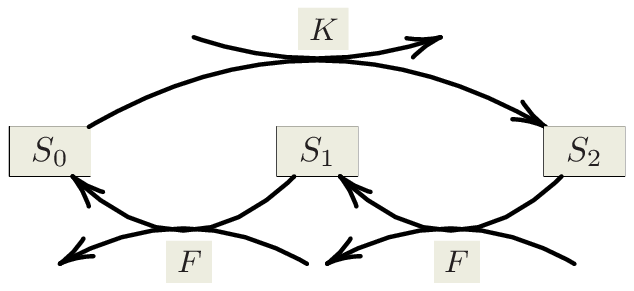}
    \subcaption{
      \label{fig:cycles-pd}
      Processive phosphorylation, distributive dephosphorylation
    }    
  \end{subfigure}
  \hfill
  \begin{subfigure}{0.3\textwidth}
    \includegraphics[width=0.9\textwidth]{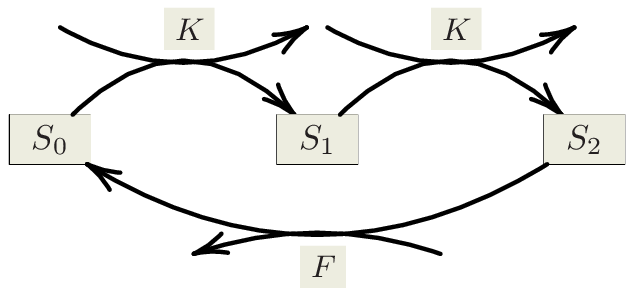}
    \subcaption{
      \label{fig:cycles-dp}
      Distributive phosphorylation, processive dephosphorylation
    }    
  \end{subfigure}

  \caption{
    Phosphorylation cycles describing the phosphorylation of $S$ at
    two binding sites. $S_0$ unphosphorylated $S$, $S_1$ and $S_2$
    mono- and bi-phosphorylated $S$. Only in the distributive
    (de)phosphorylation steps mono-phosphorylated $S_1$ is released. 
  }
  \label{fig:cycles}
\end{figure}

In contrast purely processive phosphorylation cycles have a unique
stable steady state \cite{processive-Nsite}, while the mixed cycles
depicted in Fig.~\ref{fig:cycles-pd} and \ref{fig:cycles-dp}
have a unique steady state that might not be stable, and admit oscillations
\cite{SK,shiu:hopf}. Distributive steps therefore seem to be involved
with the emergence of oscillations, in particular as in more involved
combinations of distributive and processive steps, oscillations have
been reported  as well
\cite{conradi-shiu-review,obatake2019oscillations}.

Interestingly, oscillations have not been reported for the cycle
depicted in Fig.~\ref{fig:cycles-dd}, despite considerable effort by
different research groups. One avenue to establish oscillations is to
determine values of rate constants and concentration variables where a
supercritical Hopf bifurcation occurs, see e.g.\
\cite{Kuznetsov:1995:EAB}.  In \cite{osc-010} Hopf bifurcation points
have been located in the parameter space of a variety of models from
the systems biology literature. 
In \cite[Section~5.36]{osc-010} a
simplification of the mass action model of Fig.~\ref{fig:cycles-dd} is
examined. The authors provide steady state  concentration values and rate
constants of a candidate Hopf bifurcation point. (The
rate constants are provided indirectly since they can be computed from
the  \lq convex parameters\rq{} the authors use, cf.\
Section~\ref{sec:convex-para}). 
  Applying Theorem~\ref{thm:hopf}, however, shows that at this point no Hopf 
  bifurcation occurs: the candidate point satisfies the conditions on the Hurwitz determinants but
  it fails the condition on the constant coefficient of the
  characteristic polynomial (in our case the last nonzero coefficient
  of the characteristic polynomial). Thus it cannot be a point of
  Hopf bifurcation, for more details see the discussion in
  Section~\ref{sec:discussion}. 
In summary, to the best of our knowledge, for the mass action model of
the phosphorylation cycle in Fig.~\ref{fig:cycles-dd} nor
simplifications of it, neither Hopf bifurcations nor oscillations have
been reported to date.

In this paper we work towards solving the problem of  existence or
non-existence  of Hopf bifurcations for the double phosphorylation
network. We follow the strategy outlined in \cite{shiu:hopf} and 
as in \cite{osc-010} we use convex parameters (see
Section~\ref{sec:convex-para}). As in \cite{osc-010} we consider
simplifications of the mass action model derived from
Figure~\ref{fig:cycles-dd}. Specifically we consider all four mass
action networks derived from Fig.~\ref{fig:cycles-dd} that contain
only two (out of four possible) enzyme-substrate complexes. 
This is done for two reasons: first, the four networks are
biochemically interesting and hence worth studying by
themselves (cf.\ Remark~\ref{rem:relation-MM}). Second, while it is
trivial to see that three of the 
models do not admit Hopf bifurcations, the fourth model displays a
nice structure that can be explored, even when all parameters treated
as unknown, by performing symbolic computations. Treatment of the full
model is currently out of reach due to the computational complexity.

The paper is organized as follows: in Section~\ref{sec:notation} we
introduce the notation. We further recall convex
parameters and a criterion for purely imaginary roots in 
Section~\ref{sec:background}. In Section~\ref{sec:maa-model} we
motivate the mass action models studied in this paper. In
Section~\ref{sec:Absence} we prove that for none of the four mass
action models Hopf bifurcations are possible: for each mass action
model we show that there do not exist parameter values such that a
Hopf bifurcation occurs (this is Theorem~\ref{thm:main}). The proof of the theorem relies on large symbolic computations that we present in the Maple Supplementary File  {\bf `SupplMat1.mw'} (see also `SupplMat1.pdf' for a pdf version). 
In Section~\ref{sec:discussion} we comment in more detail on the
candidate point for Hopf bifurcation presented in \cite{osc-010}.

\section{Notation}
\label{sec:notation}

We consider systems of $n$ chemical species $A_1$, \ldots, $A_n$ and
$r$ chemical reactions of the form 
\begin{displaymath}
  \sum_{i=1}^n \alpha_{ij} A_i \xrightarrow{\k_j } \sum_{i=1}^n
  \beta_{ij} A_i , \quad j=1,2,\ldots r,
\end{displaymath}
where the integer numbers $\alpha_{ij} \geq 0$, $\beta_{ij}\geq 0$
are the stoichiometric coefficients and $\k_j>0$ the rate constants. 
We use $x_i$   to denote the concentration of species $A_i$. 
Throughout this paper  we will assume that all reactions are governed
by mass action kinetics, that is, the reaction rate is proportional to
the product of the concentrations of the reacting species raised to
the power of their respective molecularities. Then the reaction rate
$v_j(\k,x)$ of the  $j$-th reaction is
\begin{equation}
  \label{eq:massaction}
  v_j(\k,x)  =\k_j \prod_{i=1}^n  x_i^{\alpha_{ij} }.
\end{equation}
With this, the above reaction network defines the following system of ordinary
differential equations
\begin{equation}
  \label{eq:mass-act}
  \dot x = N v(\k,x),
\end{equation}
where $N$ is the stoichiometric matrix. Here $N$ is of dimension $n \times r$ and $v(\k,x)$ of
dimension $r \times 1$. The $ij$-th entry of $N$  is given by the
difference of the stoichiometric coefficients:
\begin{displaymath}
  N_{ij} = \beta_{ij} - \alpha_{ij}.
\end{displaymath}

Let $\diag(\kappa)$ be the diagonal $r \times r$ matrix of rate constants
\begin{displaymath}
  \diag (\k) = \diag(\k_1,\k_2, \ldots ,\k_r)
\end{displaymath}
where the  $\k_i$ coordinate of an  $r$-dimensional vector $\k$ is the
$[i,i]$ entry of $\diag(\k)$. Let $Y$ be the $n \times r$ matrix
whose column vectors $y_j$ contain 
the stoichiometric coefficients $\alpha_{ij}$ of the reactant of 
the $j$-th reaction. The matrix $Y$ is sometimes called the
\emph{kinetic order matrix}. 
Given vectors $x,y\in \R^n$, we use the notation
$x^y=\prod_{i=1}^n x_i^{y_i}$. Then the columns of $Y$ define
the monomial vector
\begin{displaymath}
  \psi(x) =
  \begin{pmatrix}
    x^{y_1} \\ \vdots \\ x^{y_r}
  \end{pmatrix}
\end{displaymath}
and $v(\k,x)$ can be written as the product
\begin{displaymath}
  v(\k,x) = \diag(\k) \psi(x).
\end{displaymath}

For example the reaction network
\begin{displaymath}
  S_0 + K \ce{<=>[\k_1][\k_2]} S_0K \ce{->[\k_3]} S_1+K
\end{displaymath}
consists of the three reactions
\begin{displaymath}
  S_0 + K \ce{->[\k_1]} S_0 K, \quad S_0 K\ce{->[\k_2]} S_0 + K\quad  \text{ and }\quad 
  S_0 K\ce{->[\k_3]} S_1 + K 
\end{displaymath}
among the four species $S_0$, $K$, $S_1$ and  $S_0K$ taken in that
order as species $A_1$ through $A_4$. 
The stoichiometric coefficients of the first reaction are
$\alpha_{11}=\alpha_{21}=\beta_{41}=1$ and $\alpha_{31}=\alpha_{41}=
\beta_{11} = \beta_{21}= \beta_{31}=0$. The differences
$\beta_{i1}-\alpha_{i1}$ define the first column of the following 
stoichiometric matrix (the remaining columns are defined in a similar
way): 
\begin{displaymath}
  N=\left[
    \begin{array}{rrr}
      -1 &  1 &  0 \\
      -1 &  1 &  1 \\
       0 &  0 &  1 \\
       1 & -1 & -1 
    \end{array}
  \right].
\end{displaymath}
With mass action kinetics, the reaction rate vector is
\begin{displaymath}
  v(k,x) =
  \begin{pmatrix}
    \k_1 x_1 x_2 \\ \k_2 x_4 \\ \k_3 x_4
  \end{pmatrix} =  \left[
    \begin{array}{ccc}
      \k_1 & 0 & 0 \\
      0 & \k_2 & 0 \\
      0 & 0 & \k_3
    \end{array}
  \right]   \begin{pmatrix}
 x_1 x_2 \\   x_4 \\  x_4
  \end{pmatrix}.
\end{displaymath}
The kinetic order matrix is
\begin{displaymath}
  Y = \left[
    \begin{array}{ccc}
      1 & 0 & 0 \\
      1 & 0 & 0 \\
      0 & 0 & 0 \\
      0 & 1 & 1
    \end{array}
  \right].
\end{displaymath}

\section{Background}
\label{sec:background}

\subsection{A theorem to preclude Hopf bifurcations}
\label{sec:yangs-theorem}

In this subsection we state a criterion in terms of the principal minors of the Hurwitz matrix 
(Proposition~\ref{prop:yang}, Theorem~\ref{thm:hopf})  that we will use to exclude Hopf
bifurcations. The criterion follows from the results in
\cite{deciding:kahoui}. Related results for asserting Hopf
bifurcations can be found in \cite{BifTheo-009,osc-011}.   

For ease of notation consider
an ODE system parametrized by a single parameter  $\mu \in \mathbb{R}$: 
\begin{align*}
  \dot x ~=~ g_{\mu}(x)~,
\end{align*}
where $x \in \mathbb{R}^s$, and $g_{\mu}(x)$ varies smoothly in
$\mu$ and $x$. Let  $x^* \in \mathbb{R}^s$ be  a steady state of the
ODE system  for some fixed value $\mu_0$, that is,
$g_{\mu_0}(x^*)=0$. Furthermore,  we assume that we have a smooth
curve of steady states around $\mu_0$:
\begin{align}    \label{eq:curve}
  \mu ~\mapsto ~ x(\mu)~
\end{align}
that is, $g_{\mu}\left( x(\mu) \right)= 0$ for all $\mu$ close enough to $\mu_0$  such that
$x(\mu_0)=x^*$. By the Implicit Function Theorem, this curves exists
if the Jacobian of $g_{\mu_0}(x)$ evaluated at $x^*$ is non-singular.

Let $J(x(\mu),\mu)$ be the Jacobian of $g_\mu(x)$ evaluated at
$x(\mu)$. If, as $\mu$ varies, a single complex-conjugate pair of
eigenvalues of $J(x(\mu),\mu)$ crosses the imaginary axis while all
other eigenvalues remain with nonzero real parts, then  a {\em
  simple Hopf bifurcation} occurs at $(x(\mu),\mu)$.
In this case, a limit cycle arises. If the Hopf bifurcation  is
supercritical and all other eigenvalues have negative real part, then
stable periodic solutions are generated for nearby parameter values.

\begin{remark}
  \label{rem:necessary-condi}
  As we pointed out in the informal discussion above, a simple Hopf
  bifurcation requires that exactly one pair of complex conjugate
  eigenvalues crosses the imaginary axis when the parameter $\mu$
  varies. Thus, in particular, there must exist a value $\mu_0$ with
  steady state $x(\mu_0)$,
  where the Jacobian $J(x(\mu_0),\mu_0)$ has exactly one pair of
  purely imaginary eigenvalues. More generally, a Hopf
  bifurcation necessitates $J(x(\mu_0),\mu_0)$, for some value
  $\mu_0$, to have a pair of purely imaginary eigenvalues.
\end{remark}

The following proposition gives necessary and sufficient conditions
for the existence of exactly one pair of purely imaginary
eigenvalues in the scenario we encounter later on. It is
based on Hurwitz matrices, which we define first:

\begin{definition} \label{def:hurwitz}
  The {\em $i$-th Hurwitz matrix} of a univariate polynomial 
  $p(z)= a_0 z^s + a_{1} z^{s-1} + \cdots + a_s$  
  is the following $i \times i$ matrix:
  \[
    H_i ~=~ 
    \begin{pmatrix}
      a_1 & a_0 & 0 & 0 & 0 & \cdots & 0 \\
      a_3 & a_2 & a_1 & a_0 & 0 & \cdots & 0 \\
      \vdots & \vdots & \vdots &\vdots & \vdots &  & \vdots \\
      a_{2i-1} & a_{2i-2} & a_{2i-3} & a_{2i-4} &a_{2i-5} &\cdots & a_i
    \end{pmatrix}~,
  \]
  in which the $(k,l)$-th entry is $a_{2k-l}$ 
  as long as $0\leq 2 k - l \leq s$, and
  $0$ otherwise.  Note $H_s=a_s H_{s-1}$.
\end{definition}

Returning to the ODE system $\dot{x} = g_\mu(x)$, consider the
characteristic polynomial of the Jacobian matrix
$J(x(\mu),\mu)$
\begin{align*}
  p_{\mu}(z) 
  ~:=~ \det \left(z I - J(x(\mu),\mu) \right)
  ~=~  a_0(\mu) z^s + a_{1}(\mu) z^{s-1} + \cdots + a_s(\mu)~,
\end{align*}
where the determinant of a matrix $A$ is denoted by $\det A$.
For $i=1,\dots, s$, we let $H_i(\mu)$ be the  $i$-th Hurwitz
matrix of $p_{\mu}(z)$.

We now state the criterion we will use to determine whether the
characteristic polynomial has a pair of purely imaginary roots.

\begin{proposition} \label{prop:yang}
  In the setup above, let $s\geq 2$, $\mu_0$ be fixed, and let
  $H_i(\mu_0)$ denote the $i$-th Hurwitz matrix of $p_{\mu_0}(z)$. We
  assume that 
  \begin{equation}\label{eq:Hpos}
    \det H_1(\mu_0)>0,  \  \dots,  \ \det H_{s-2}(\mu_0)>0.
  \end{equation}
  Then $p_{\mu_0}(z)$ has at most one pair of symmetric roots, and has
  exactly one if and only if  $\det H_{s-1}(\mu_0) = 0$. In this
  case, the pair consists of purely imaginary roots if and only if
  $a_s(\mu_0)>0$.
\end{proposition}
\begin{proof}
  The first statement follows from Corollary 3.2  in
  \cite{deciding:kahoui} using \eqref{eq:Hpos}. So, assume
  $p_{\mu_0}(z) $ has a pair of symmetric roots $z,-z$ such that
  $\det H_{s-1}(\mu_0)=0$. By Lemma 3.3 in \cite{deciding:kahoui},
  the Routh-Hurwitz criterion for stable polynomials, and using
  \eqref{eq:Hpos}, we conclude that all the other $s-2$ roots of
  $p_{\mu_0}(z) $  have negative real part.  Now $a_s(\mu_0)$ is
  $(-1)^s$ times the product of the roots of $p_{\mu_0}(z) $, and
  hence its sign agrees with the sign of $(-1)^s(-1)^{s-2}z(-z)=
  -z^2$. We conclude that the pair of symmetric roots of
  $p_{\mu_0}(z) $ is real if   $a_s(\mu_0)\leq 0$ and purely
  imaginary if $a_s(\mu_0)> 0$.
\end{proof}

As a consequence of Propostion~\ref{prop:yang}, we obtain the
following criterion to preclude Hopf bifurcations:
\begin{theorem}\label{thm:hopf}
  Consider the dynamical system $\dot x = g(x)$ with $s\geq 2$ and
  assume there exists a curve of steady states $x(\mu)$ as in the
  above setting.
  As before, let $p_{\mu}(z)$ be the characteristic polynomial of
  the Jacobian $J(x(\mu),\mu)$ of $g_{\mu}(x)$ evaluated at
  $x(\mu)$. Further let $H_i(\mu)$ denote the $i$-th Hurwitz matrix
  of $p_{\mu}(z)$. \\
  If 
  \[
    \det H_1(\mu)>0,  \  \dots,  \ \det
    H_{s-2}(\mu)>0
  \]
  for all $\mu$ and
  either
  \begin{align*}
    a_s(\mu) \leq 0  \text{ whenever } \det H_{s-1}(\mu) = 0,
    \intertext{or}
    \det H_{s-1}(\mu) \neq 0 \text{ for all $\mu$,}
  \end{align*}
  then the system does not undergo a (simple) Hopf bifurcation.
\end{theorem}
\begin{proof}
  Follows from Remark~\ref{rem:necessary-condi} and
  Proposition~\ref{prop:yang}.
\end{proof}

In the following sections we will use Theorem~\ref{thm:hopf} to prove
the non-existence of Hopf bifurcations in subnetworks of the mass
action network derived from Fig.~\ref{fig:cycles-dd}. For a reaction
network with $n$ species, if the rank $s$ of the stoichiometric matrix
$N$ is not maximal, as it is the case for our networks, then the
dynamics takes place in the invariant $s$-dimensional  linear
subspaces $x_0+\im N$. This implies that $0$ is a root of the
characteristic polynomial of $N v(\k,x)$ with multiplicity $n-s$ and
hence it factors as $$ p_{\k,x}(z)= z^{n-s} \big( a_{0}(\k,x) z^s  +
a_{1}(\k,x) z^{s-1} + \cdots + a_s(\k,x) \big).$$
The polynomial $a_{0}(\k,x) z^s  + a_{1}(\k,x) z^{s-1} + \cdots +
a_s(\k,x)$ is the characteristic polynomial of the Jacobian of the
restriction of system \eqref{eq:mass-act} to $x+\im N$, and hence we
apply Theorem~\ref{thm:hopf} to this polynomial.

\subsection{Convex parameters}
\label{sec:convex-para}
By the previous subsection, in order to determine whether a Hopf
bifurcation can arise in our systems, we need to analyze the Jacobian
matrix of the right-hand side of \eqref{eq:mass-act} for all possible
values of rate constants $\k$ and positive steady states $x^*$. Here
we reparametrize the Jacobian matrices using so-called {\it convex
  parameters}. These parameters were introduced by Clarke in
\cite{bc1980} to analyze the stability of  mass action reaction systems~\eqref{eq:mass-act}, in the context of {\it Stoichiometric Network Analysis}  (SNA) theory.

Let a  positive steady state $x^*\in \R^n_{>0}$ of (\ref{eq:mass-act}) satisfy  the polynomial system $N\diag (\k) \psi (x^*)=0$. Then  the rate
functions   $v=v(\k,x^*)$ satisfy the linear 
problem 
\begin{equation}\label{eq:lin-problem}
  N v=0, \quad v \geq 0 .
\end{equation}
The vector $v$ is referred to as a {\it flux vector} in the SNA theory \cite{bc1988}. 
The solutions $v$ of (\ref{eq:lin-problem}) define a convex polyhedral  cone called  the {\it flux cone}. Convex polyhedral cones have a finite number of extreme vectors up to a  scalar positive multiplication \cite{rr2015}. Therefore,
any flux vector $v$  can be represented as a nonnegative  linear combination of its extreme vectors $\{E_1, \ldots  ,E_l  \}$
\begin{equation}\label{eq:convex-cone}
  v = \sum_{i=1}^l \lambda_i E_i = E \lambda, \quad \mbox{all}   \quad
  \lambda_i \geq 0,
\end{equation}
where $E$ is the matrix with columns $E_1,\dots,E_l$ and
$\lambda=(\lambda_1,\dots,\lambda_l)$.

\begin{remark}[cf.~\cite{rr2015}]
  \begin{itemize} 
  \item[(a)] The vectors $E_1, E_2, \ldots, E_l$ need not be linearly independent.
  \item[(b)] If all extreme vectors $ E_1, \ldots , E_l$ are unit vectors,
    then their choice is unique.  
  \item[(c)] 
      When $v=v(\k,x^*)$ and $x^*\in \R^n_{>0}$, then all components
      of $v$ are positive. This might impose some restrictions on the
      possible values of $\lambda$ in \eqref{eq:convex-cone}.
  \end{itemize}
\end{remark}

The  nonnegative coefficients
$ \lambda_1, \ldots , \lambda_l$ 
in (\ref{eq:convex-cone}) are often referred to as  convex parameters
in the literature.  However, they  do not account alone for all new
parameters -  the  other group of  parameters used in  SNA theory
are reciprocals of each positive steady state coordinate
$x_k^*>0$. They are  denoted by  
\begin{equation}\label{eq:h}
h_k=\frac{1}{x_k^*} , \quad k=1,\ldots , n.
\end{equation}

\begin{definition}
  A vector of  convex parameters  is a vector of the form
  $(h,\lambda)=(h_1,\dots,h_n,\lambda_1,  \ldots , \lambda_l)\in
  \R^n_{>0} \times \R^l_{\geq 0}$ such that $E\lambda\in \R^r_{>0}$. 
\end{definition}
 
The convex parameters are convenient for parameterizing the Jacobian
$J(\k,x)$ evaluated at a positive steady state $x=x^*$. To see this, 
note that the $(j,i)$-th entry of the Jacobian  of $v(\k,x)$ evaluated at $x^*$ is
\begin{displaymath}
  {\frac{\partial\, v_j (\k,x)}{\partial \, x_i}}_{\lvert x=x^*}
  =\frac{\alpha_{ij}v_j(\k,x^*)}{x_i^*}=\alpha_{ij} v_j(\k,x^*) \tfrac{1}{x^*_i} .
\end{displaymath}
 Hence,  the Jacobian   of $N v(\k,x)$ evaluated at $x^*$ is
 \[ 
   J(\k,x)_{|x=x^*}= N \diag( v(\k,x^*)) \, Y^T \diag(\tfrac{1}{x^*}),
 \]
 where we use the vector notation
 \[
   \tfrac{1}{x} = \big(\, \tfrac{1}{x_1}, \ldots ,\tfrac{1}{x_n}\, \big)^T.
 \]

 Using now  \eqref{eq:convex-cone} and \eqref{eq:h} to write $v(\k,x^*)
 = E\lambda$, the Jacobian of $Nv(\k,x)$ evaluated at  $x^*$ can be
 written as  
 \begin{equation}
   \label{eq:new-jac}
   J(\k,x)_{|x=x^*}= J(h,\lambda) = N \diag(E\lambda) Y^T
   \diag(h),
\end{equation}
with $(h,\lambda)$ a vector of convex parameters. 
Therefore, given a vector of rate constants $\k$ and a corresponding
positive steady state $x^*$, there exist convex parameters
$(h,\lambda)$ such that  equality \eqref{eq:new-jac} holds.

Conversely, given convex parameters $(h,\lambda)$, we define $x^*=1/h$
and let \[ \k = \diag(\psi(h))E\lambda,\]
which is a positive vector since all entries of $E\lambda$  are
positive. Then, using that $\psi_j(x)^{-1} =  \psi_j(x^{-1})$, we
obtain $v(\k,x^*)=\diag(\k) \psi(x^*)=E\lambda$, and equality
\eqref{eq:new-jac} holds as well. This proves the following 
proposition:

\begin{proposition}\label{prop:correspondence}
The set of Jacobian matrices  $J(\k,x^*)$ for all $\k$ and
corresponding positive steady states $x^*$ agrees with the set
of matrices defined by the right-hand side of \eqref{eq:new-jac}, for
all $h\in \R^n_{>0}$ and $\lambda\in \R^l_{\geq 0}$ such that
$E\lambda\in \R^r_{>0}$.
\end{proposition}

The computation of the Jacobian in convex parameters
\eqref{eq:new-jac} appears in  great detail  in previous works
\cite{ac87,ges05}.
In \cite{conradi-switch} it is used to detect saddle-node
bifurcations.
In Section~\ref{sec:Absence}, we use the Jacobian in convex
coordinates given in \eqref{eq:new-jac} and apply Theorem~\ref{thm:hopf} to
conclude that there does not exist a point ($h$,$\lambda$) where a
Hopf bifurcation occurs. Using the equality between the two
sets of matrices in  Proposition~\ref{prop:correspondence}, this
will imply that there do not exist $\k$  and $x^*$ where a Hopf
bifurcation occurs.

\begin{remark}
  The coefficients of the characteristic polynomial of $J(\k,x^*)$ are
  homogeneous  polynomials in the convex parameters if the Jacobian
  matrix is parametrized as in \eqref{eq:new-jac}. 
\end{remark}

\section{The mass action model derived from Fig.~\ref{fig:cycles-dd} 
  and its simplifications}
\label{sec:maa-model}

Figure~\ref{fig:cycles-dd} contains four phosphorylation events: the
phosphorylation of $S_0$ and $S_1$ and the dephosphorylation of $S_2$
and $S_1$. At the level of mass action kinetics each of these
phosphorylation events can be  described by the following reactions
(with $i=0,1$):
\begin{displaymath}
  S_i + K \ce{<=>} KS_i \ce{-> } S_{i+1}+K \quad \text{ and }\quad  S_{i+1} + F
  \ce{<=> } FS_{i+1} \ce{-> } S_i+F.
\end{displaymath}
Consequently, if all phosphorylation events depicted in
Fig.~\ref{fig:cycles-dd} are described at the mass action level one
obtains the following reaction network: 
\begin{equation}
  \tag{$\mathcal{N}$}
  \label{eq:network}
  \begin{split}
    S_0 + K \ce{<=>[\k_1][\k_2]} KS_0 \ce{->[\k_3]} S_1+K
    \ce{<=>[\k_4][\k_5]} KS_1 \ce{->[\k_6] } S_2+K \\
    S_2 + F  \ce{<=>[\k_7][\k_8]} FS_2 \ce{->[\k_9]} S_1+F
    \ce{<=>[\k_{10}][\k_{11}]} FS_1 \ce{->[\k_{12}]} S_0+F.
  \end{split}
\end{equation}
To apply Theorem~\ref{thm:hopf}, one has to compute Hurwitz determinants (see
Section~\ref{sec:yangs-theorem} above). These are determinants of
matrices that are composed of the coefficients of the characteristic
polynomial of the Jacobian of the ODEs defined by~\ref{eq:network}. 
In order to show whether or not there exist
some values of the rate constants where a Hopf bifurcation occurs, we
have to treat all rate constants as fixed but unknown. 
The coefficients of the characteristic polynomial may contain several 
hundred terms (cf.\ the supporting information of
\cite{conradi-mincheva}). 
To facilitate the analysis we consider the following simplifications
of~\ref{eq:network}:
\begin{enumerate}[{(}i{)}]
\item We consider only the forward reaction of the reversible
  reactions
  \begin{displaymath}
    S_i + K \ce{<=>} KS_i \quad \text{ and } \quad S_{i+1} + F \ce{<=> } FS_{i+1}.
  \end{displaymath}
This is a reasonable assumption if the rate constants for  the
reversible reactions are small.
\item We consider only two of the four enzyme-substrate complexes
  $KS_0$, $KS_1$, $FS_2$ and $FS_1$.  
\end{enumerate}
There are six ways to choose two complexes out of four. Due to the
symmetry of the ODE system obtained by interchanging $K$ and $F$,
$S_0$ and $S_2$, and relabeling the rate constants as appropriate,  it
suffices to consider the following four simplified networks derived
from~\ref{eq:network}: 
\begin{itemize}
\item The network containing only $KS_0$ and $FS_2$:
  \begin{equation}
    \tag{$\mathcal{N}_1$}
    \label{eq:N1}
    \begin{split}
      K + S_0 \ce{->[\k_1]} KS_0 \ce{->[\k_2]} K+S_1  \ce{->[\k_3]}
      K+S_2 \\ 
      F + S_2 \ce{->[\k_4]} FS_2 \ce{->[\k_5]} F+S_1  \ce{->[\k_6]}
      F+S_0.
    \end{split}
  \end{equation}
\item The network containing only $KS_0$ and $FS_1$
  (mathematically equivalent to the network containing only
  $FS_2$ and $KS_1$):
  \begin{equation}
    \tag{$\mathcal{N}_2$}
    \label{eq:N2}
    \begin{split}
      K + S_0 \ce{->[\k_1]} KS_0 \ce{->[\k_2]}   K+S_1  \ce{->[\k_3]}
      K+S_2 \\
      F + S_2 \ce{->[\k_4]} F+S_1 \ce{->[\k_5]} FS_1  \ce{->[\k_6]}
      F+S_0.
    \end{split}
  \end{equation}
\item The network containing only $KS_0$ and $KS_1$
  (mathematically equivalent to the network containing only
  $FS_2$ and $FS_1$):
  \begin{equation}
    \tag{$\mathcal{N}_3$}
    \label{eq:N3}
    \begin{split}
      K + S_0 \ce{->[\k_1]} KS_0 \ce{->[\k_2]}  K+S_1  \ce{->[\k_3]}
      KS_1   \ce{->[\k_4]} K+S_2 \\
      F + S_2  \ce{->[\k_5]} F+S_1  \ce{->[\k_6]}  F+S_0.
    \end{split}
  \end{equation}
\item The network containing only $KS_1$ and $FS_1$:
  \begin{equation}
    \tag{$\mathcal{N}_4$}
    \label{eq:N4}
    \begin{split}
      K + S_0 \ce{->[\k_1]}K+S_1  \ce{->[\k_2]} KS_1  \ce{->[\k_3]}
      K+S_2 \\
      F + S_2  \ce{->[\k_4]} F+S_1 \ce{->[\k_5]} FS_1 \ce{->[\k_6]}
      F+S_0.
    \end{split}
  \end{equation}
\end{itemize}

\begin{remark}
  The dynamics of the network obtained upon removal of an intermediate
  resembles that of the original network when the binding reaction
  occurs at a slower time scale than unbinding reactions. Specifically,
  under this assumption, the simplified system corresponds to the slow
  system arising from Tikhonov-Fenichel singular perturbation reduction
  of the original system, which additionally agrees with its
  quasi-steady-state reduction \cite{walcher:feliu:wiuf}.
  Upon removal of a reverse unbinding reaction, the assumption is that
  this reaction occurs at a much slower time scale than the other
  reactions in the network, and hence is negligible.
\end{remark}

\begin{remark}
  \label{rem:relation-MM}
  A biochemical interpretation of the simplification leading to the
  networks \eqref{eq:N1} -- \eqref{eq:N4} can be obtained by
  comparing it to the well-known Michaelis-Menten approximation:
  we view our
  simplification as 
  similar in spirit but with different focus and hence concurrent. To
  see this recall that to simplify a reaction network based on the
  Michaelis-Menten approximation, a mass action network of the form 
  \begin{equation}   \label{eq:NE}
      S_i + K \ce{<=>} KS_i \ce{-> } S_{i+1}+K
  \end{equation}
  is replaced by a network of the form
  \begin{displaymath}
    S_i \ce{->[v_{\text{MM}}]} S_{i+1},
  \end{displaymath}
  where $v_{\text{MM}}$ is the familiar Michaelis-Menten kinetics
  \begin{displaymath}
    v_{\text{MM}}  = \frac{k_{cat}K_0 [S_i]}{\mathsf{K}_m + [S_i]},
  \end{displaymath}
  with catalytic constant $k_{cat}$, total enzyme concentration $K_0$,
  Michaelis-Menten constant $\mathsf{K}_m$, and where $[\cdot ]$
  denotes the concentration.
  This is a standard practice of model
  simplification, for example if the condition formulated by Briggs
  and Haldane holds \cite{enz-kinetics}; see
  \cite{Walcher:MM1,Walcher:MM2} for the mathematical study of this 
  reduction. 
  In \cite{osc-012} it is shown that the dynamical system that
  arises from the Michaelis-Menten approximation of the full
  network~\ref{eq:network} does not exhibit periodic orbits.

  A particular feature of the Michaelis-Menten kinetic 
  $v_{\text{MM}}$ is that it is approximately linear in 
  $[S_i]$ for small values of $[S_i]$ and approximately constant for very
  large values. Thus the simplification based on the Michaelis-Menten
  approximation covers, for example, the saturation of
  available enzyme with substrate very well. However, it does not
  capture the influence of varying concentrations of the free enzyme
  $K$. 
  
  In our simplification we replace a mass action network of the form 
  \eqref{eq:NE} by a mass action network of the form
  \begin{displaymath}
        S_i + K \ce{->[\k]} S_{i+1} + K,
  \end{displaymath}
  where, according to the law of mass action the reaction rate is
  \begin{displaymath}
    v_{\text{MA}} = \k \cdot [S_i] \cdot [K].
  \end{displaymath}
  Clearly, $v_{\text{MA}}$ is linear in $[S_i]$, hence for small values of
  $S_i$ and values of $[K]$ close to its total concentration $K_0$ the
  two reaction rates $v_{\text{MA}}$ and $v_{\text{MM}}$ behave
  similar. However, $v_{\text{MA}}$ is not able to reproduce the
  saturation of enzyme by the 
  substrate. But it does capture the influence of varying
  concentrations of free enzyme $[K]$. Hence for some values of $[K]$ and
  $[S_i]$ our simplification behaves similar to the one based on the
  Michaelis-Menten approximation, but not everywhere. Moreover, our
  simplification covers the influence of varying concentrations of
  $K$.

  In summary, 
  the simplification based on the Michaelis-Menten approximation on
  the one hand is well suited to 
  describe the saturation of the enzyme with substrate but does not
  account for the influence of varying concentrations of free
  enzyme. Our simplification on the other hand cannot account for
  enzyme saturation, but for the influence of varying concentrations
  of free enzyme. Moreover, for small values of the substrate
  concentration and for concentrations of free enzyme close to its
  total amount both simplifications behave similar. Hence we view our
  simplification as biochemically concurrent to the one based on the
  Michaelis-Menten approximation. Both simplifications fail to
  accommodate the complete behavior of the distributive and sequential double
  phosphorylation cycle of Fig.~\ref{fig:cycles-dd}. But both cover
  different but equally important aspects of its behavior and hence
  are well worth studying.
\end{remark}

\begin{remark}\label{rk:simplify}
  If any of the networks \ref{eq:N1} -- \ref{eq:N4} had a Hopf
  bifurcation giving rise to   periodic solutions,
  then by \cite{Banaji18:oscillations}, so would the full mechanism in
  \ref{eq:network}.
  In particular, by \cite{Banaji18:oscillations} the same is true if any of the
  irreversible reactions is made reversible and in particular, if 
  unbinding reactions are considered in the formation the
  complexes $K S_0$, $K S_1$, $F S_1$ or $F S_2$.
\end{remark}

\section{Absence of Hopf bifurcations}
\label{sec:Absence}

In this section we apply Theorem~\ref{thm:hopf}, using the discussion after it, to the networks
\ref{eq:N1} -- \ref{eq:N4}. To this end we use
equation~(\ref{eq:new-jac}) to determine the Jacobian matrices
$J_1(h,\lambda)$, \ldots, $J_4(h,\lambda)$ of networks \ref{eq:N1} --
\ref{eq:N4}, correspondingly. The computations are done symbolically and can
  be found in the supporting file {\bf `SupplMat1.mw'}.

We first comment on some common features of the networks: 
\begin{remark}\label{re:rem7}
  \begin{itemize}
  \item[(i)] Every network \ref{eq:N1} --
\ref{eq:N4} consists of $7$ species and $6$ reactions. 
  \item[(ii)] The stoichiometric matrix of every network has rank $s=4$, as every network has 
    three conserved quantities (the total amount of substrate, kinase
    and phosphatase).
  \item[(iii)] For every network the cone (\ref{eq:lin-problem}) is
    spanned by two nonnegative vectors $w_0$ and $w_1$ such that
    $\lambda_1 w_0 + \lambda_2 w_1$ is a positive vector if and only
    if $\lambda_1,\lambda_2>0$.
  \end{itemize}
\end{remark}
 
Observe that a scaling $\alpha \lambda$ of the vector
$\lambda=(\lambda_1,\lambda_2)$ with $\alpha>0$ translates into a
scaling $\alpha \k$ of the vector of rate constants $\k$ under the
correspondence of parameterization in
Proposition~\ref{prop:correspondence}.  Further $x^*$ is a steady
state for $\k$ if and only if it is for $\alpha \k$ and 
$\alpha J(\k,x^*)=J(\alpha \k ,x^*)$. The latter implies that any
eigenvalue of $J(\alpha \k, x^*)$ is  in fact an eigenvalue of
$J(\k, x^*)$ multiplied by $\alpha>0$. Thus,  $J(\alpha \k ,x^*)$
has a pair of purely imaginary eigenvalues if and only if  $J( \k
,x^*)$  has  a pair of purely imaginary eigenvalues.

Hence, it is enough to take one of $\lambda_1,\lambda_2$ to be one
(since both are positive), and we let the elements in the kernel of
the stoichiometric matrices $N_i$ be of the form 
\begin{displaymath}
  w_0 +\lambda  w_1,\qquad \lambda>0.
\end{displaymath}
Therefore, we consider every Jacobian matrix $J_i(h,\lambda)$ according to equation
\eqref{eq:new-jac} parametrized by $8$ parameters: the
parameter $\lambda$ and $h_1,\dots,h_7$.

By Remark~\ref{re:rem7} (i) and (ii) it follows that  the
characteristic polynomial of every 
Jacobian $J_i(h,\lambda)$ is a degree $7$ polynomial of the form 
\begin{displaymath}
  z^3\, \big(a_0(h,\lambda)z^4 + \ldots + a_3(h,\lambda ) z +  a_4(h,\lambda) \big),
\end{displaymath}
where each $a_i$ depends on the $8$ parameters $\lambda, h_1,\ldots, h_7$ (cf.\ the file {\bf `SupplMat1.mw'}). Following the discussion after Theorem~\ref{thm:hopf}, for each network we   compute  $a_0(h,\lambda),\ldots, a_4(h,\lambda)$, $\det H_1(h,\lambda)$, $\det H_2(h,\lambda)$ and $\det H_3(h,\lambda)$ (cf.\ Definition~\ref{def:hurwitz}) and show the following proposition:

\begin{proposition}\label{prop:analysis}
With the notation above, we have
\begin{enumerate}[(i)]
\item Network~\ref{eq:N1}: $\det H_1(h,\lambda)$ and $\det H_2(h,\lambda)$
  contain only positive monomials, $a_4(h,\lambda)$ and $\det H_3(h,\lambda)$
  contain monomials of both signs. But $\det H_3(h,\lambda)$ is positive
  whenever $a_4(h,\lambda)>0$ (Proposition~\ref{lem:no-hopf-N1} in
  Section~\ref{sec:network-N1}).
\item Network~(\ref{eq:N2}) and (\ref{eq:N3}): $\det H_1(h,\lambda)$,
  $\det H_2(h,\lambda)$ and $\det H_3(h,\lambda)$ contain only positive
  monomials, $a_4(h,\lambda)$ contains monomials of both signs. 
    Thus $\det H_i(h,\lambda)>0$, $i=1,2,3$ for positive $h$ and 
    $\lambda$ and in particular, $\det H_3(h,\lambda)\neq 0$.
\item Network~\ref{eq:N4}: $\det H_1(h,\lambda)$, $\det H_2(h,\lambda)$ and
  $\det H_3(h,\lambda)$ and $a_4(h,\lambda)$ contain only positive
  monomials.
    As in case of \ref{eq:N2} and \ref{eq:N3}, one has $\det
    H_i(h,\lambda)>0$, $i=1,2,3$ for positive $h$ and $\lambda$ and 
    in particular, $\det H_3(h,\lambda)\neq 0$. 
\end{enumerate} 
\end{proposition}

In particular, this proposition (which is proven below) tells us that in all four networks,  $\det H_3(h,\lambda)\neq 0$ whenever $a_4(h,\lambda)>0$.
As a consequence we obtain the following theorem.  

\begin{theorem}\label{thm:main}
  For the networks \ref{eq:N1} -- \ref{eq:N4} there do not exist
  rate constants $\k$ and a corresponding positive steady state $x^*$
  such that the Jacobian matrix  $J_i(\k,x^*)$ has a pair of purely
  imaginary eigenvalues. Thus, in particular, there do not exist $\k$
  and $x^*$ where a Hopf bifurcation occurs.  
\end{theorem}
\begin{proof}
  By Proposition~\ref{prop:analysis} and the correspondence between the two parameterization of the Jacobian given in   Proposition~\ref{prop:correspondence}, there do not exist  $\k$ and a corresponding positive steady state $x^*$  such that the corresponding Hurwitz determinant $\det H_3(\k,x^*)$ vanishes and  the coefficient of lowest degree of the characteristic polynomial $a_4(\k,x^*)$ is positive. 
Further, for all networks  $\det H_1(\k,x^*)>0$   and $\det H_2(\k,x^*)>0$ for all $\k,x^*$. Hence, by 
Theorem~\ref{thm:hopf}, the Jacobian matrix does not have a pair of purely imaginary eigenvalues. 
\end{proof}

\begin{remark}
By \cite{sadeghimanesh:multi}, the networks $\mathcal{N}_1$, $\mathcal{N}_2$ and $\mathcal{N}_3$ are multistationary, while $\mathcal{N}_4$ is not. Hence, by  \cite{FeliuPlos}, the coefficient of lowest degree of the Jacobian must vanish for some $\k$ and positive steady state $x^*$,  and consistently $a_4(h,\lambda)$ must contain monomials of both signs. 
\end{remark}

\begin{remark}
  We have focused on networks with two intermediates, as network
  (\ref{eq:N1}) provides the first non-trivial case and might shed
  light on how to approach the full network. All simplifications of
  the full network \eqref{eq:network} obtained after removing three
  intermediates, that is, with only one intermediate left, satisfy
  that $\det H_{s-1}(h,\lambda)$ is a sum of positive terms and hence
  does not vanish. Consequently, Hopf bifurcations do not arise.
\end{remark}

All that remains is to show Proposition~\ref{prop:analysis}. This is   done through 
mathematical reasoning aided by symbolic computations performed in Maple and Mathematica. 
In the following subsections we present for each network the
stoichiometric matrix $N_i$, the kinetic order matrix $Y_i$, the
matrix $E_i$ whose columns are vectors $w_0$ and $w_1$ that generate
the cone (\ref{eq:lin-problem}) and the Jacobian matrix
$J_i(h,\lambda)$.
For the networks~(\ref{eq:N1}) -- (\ref{eq:N4}) we found these vectors
simply by finding a basis of the kernel of $N$, which has dimension
$2$, and deriving extreme vectors. For larger networks software like
CellNetAnalyzer \cite{Kamp2017} can be used.
Where appropriate we then discuss the coefficient $a_4(h,\lambda)$ and
the determinants of the Hurwitz matrices $\det H_i(h,\lambda)$. The
computations are  given in the supplementary file {\bf
  `SupplMat1.mw'}.

\subsection{Network~\ref{eq:N1}}
\label{sec:network-N1}

We denote by $x_1,x_2,x_3,x_4,x_5,x_6,x_7$ the concentrations of $K$,
$F$, $S_0$, $S_1$, $S_2$, $KS_0$, $FS_2$ respectively.
Then the stoichiometric matrix, the kinetic order matrix and
the matrix $E_1$ are
\begin{displaymath}
  N_1 = \left[
    \begin{array}{rrrrrr}
      -1 & 1 & 0 & 0 & 0 & 0 \\
      0 & 0 & 0 & -1 & 1 & 0 \\
      -1 & 0 & 0 & 0 & 0 & 1 \\
      0 & 1 & -1 & 0 & 1 & -1 \\
      0 & 0 & 1 & -1 & 0 & 0 \\
      1 & -1 & 0 & 0 & 0 & 0 \\
      0 & 0 & 0 & 1 & -1 & 0 \\
    \end{array}
  \right], 
  Y_1 = \left[
    \begin{array}{cccccc}
      1 & 0 & 1 & 0 & 0 & 0 \\
      0 & 0 & 0 & 1 & 0 & 1 \\
      1 & 0 & 0 & 0 & 0 & 0 \\
      0 & 0 & 1 & 0 & 0 & 1 \\
      0 & 0 & 0 & 1 & 0 & 0 \\
      0 & 1 & 0 & 0 & 0 & 0 \\
      0 & 0 & 0 & 0 & 1 & 0 \\
    \end{array}
  \right] \text{ and } 
  E_1 = \left[
    \begin{array}{cc}
      1 & 0 \\
      1 & 0 \\
      0 & 1 \\
      0 & 1 \\
      0 & 1 \\
      1 & 0 \\
    \end{array}
  \right].
\end{displaymath}
Using convex parameters,   the Jacobian matrix in terms $h_1,\dots,h_7,\lambda$ as  given in \eqref{eq:new-jac} is \begin{displaymath}
  J_1(h,\lambda) = 
  \left[
    \begin{array}{ccccccc}
      -h_{1} & 0 & -h_{3} & 0 & 0 & h_{6} & 0 \\
      0 & -h_{2} \lambda & 0 & 0 & -h_{5} \lambda & 0 & h_{7} \lambda \\
      -h_{1} & h_{2} & -h_{3} & h_{4} & 0 & 0 & 0 \\
      -h_{1} \lambda & -h_{2} & 0 & h_{4} (-1-\lambda) & 0 & h_{6} & h_{7} \lambda \\
      h_{1} \lambda & -h_{2} \lambda & 0 & h_{4} \lambda & -h_{5} \lambda & 0 & 0 \\
      h_{1} & 0 & h_{3} & 0 & 0 & -h_{6} & 0 \\
      0 & h_{2} \lambda & 0 & 0 & h_{5} \lambda & 0 & -h_{7} \lambda \\
    \end{array}
  \right].
\end{displaymath}
We observe that the coefficient $a_4(h,\lambda)$ of the characteristic
polynomial contains monomials of both signs. 
We compute the associated Hurwitz determinants $\det H_1(h,\lambda)$,
$\det H_2(h,\lambda)$ and $\det H_3 (h,\lambda)$ and obtain that $\det H_1(h,\lambda)$ and
$\det H_2(h,\lambda)$ are sums of positive monomials and that $\det H_3(h,\lambda)$ contains monomials of both signs as well. Hence both
$a_4(h,\lambda)$ and $\det H_3(h,\lambda)$ 
contain monomials of both signs and can potentially be zero.

In the remainder of this section we prove the following:
\begin{proposition}
  \label{lem:no-hopf-N1}
  Consider the coefficient $a_4(h,\lambda)$ and the Hurwitz
  determinant $\det H_3(h,\lambda)$ of network ~\ref{eq:N1} and given in the file  {\bf `SupplMat1.mw'}. Then
  \begin{displaymath}
\mbox{if} \quad    a_4(h,\lambda) > 0, \quad \text{then}\quad \det H_3(h,\lambda) >0.
  \end{displaymath}
\end{proposition}
\begin{proof}
The coefficient $a_4(k,\lambda) $  of the Jacobian is 
\begin{multline*}
 -\lambda^{2} \big( h_{1}h_{2}h_{3}h_{4}+h_{1}h_{2}h_{3}h_{5}+h_{1}h_{2}h_{4}h_{5}+h_{1}h_{3}h_{4}h_{5}+h_{1}h_{3}h_{4}h_{7}-h_{1}h_{4}h_{5}h_{7}+h_{2}h_{3}
h_{4}h_{5} \\ -h_{2}h_{3}h_{4}h_{6}+h_{2}h_{4}h_{5}h_{6}-h_{3}h_{4}h_{5}h_{6}-h_{3}h_{4}h_{5}h_{7}-h_{3}h_{4}h_{6}h_{7}-h_{3}h_{5}h_{6}h_{7}-h_{4}h_{5}h_{6}h_{7} \big).
\end{multline*}
Since $\lambda^2$ factors out and it does not affect the sign, we consider
\begin{multline*} c_0 := - h_{1}h_{2}h_{3}h_{4}-h_{1}h_{2}h_{3}h_{5}-h_{1}h_{2}h_{4}h_{5}-h_{1}h_{3}h_{4}h_{5}-h_{1}h_{3}h_{4}h_{7}+h_{1}h_{4}h_{5}h_{7}-h_{2}h_{3}h_{4}h_{5} \\ + h_{2}h_{3}h_{4}h_{6} - h_{2}h_{4}h_{5}h_{6} + h_{3}h_{4}h_{5}h_{6} + h_{3}h_{4}h_{5}h_{7} + h_{3}h_{4}h_{6}h_{7} + h_{3}h_{5}h_{6}h_{7} + h_{4}h_{5}h_{6}h_{7}.
\end{multline*}

We show that $c_0>0$ implies $\det H_3>0$, omitting the argument of $H_3$. The computations are performed in Maple, but we explain here the computational procedure for the proof.

We start by noting that $c_0$ can be written as:
$$b_0 =  ( h_{6}-h_{1} )( h_{2}+h_{5}+h_{7}) h_{3}h_{4} + ( h_{7}-h_{2}) ( h_{1}+h_{3}+h_{6} )
 h_{4}h_{5} +( h_{6}h_{7}-h_{1}h_{2} ) h_3h_{5} .$$
We see immediately that if $h_1>h_6$ and $h_2>h_7$, then $c_0<0$. We do not need to study this case. 

We consider the case $h_6\geq h_1$ and $h_7\geq h_2$, such that one of the two inequalities is strict, otherwise $c_0=0$. 
 In this case $c_0\geq 0$. 
We introduce new nonnegative  parameters $v_1,v_2$ and substitute $h_6=h_1+v_1$ and $h_7=h_2+v_2$. This encodes the inequalities. We perform this substitution into $\det H_3$ using Maple, expand the new polynomial, and check the sign of the coefficients. All coefficients in $\det H_3$ are positive, meaning that $\det H_3$ will be positive in this case. This holds if $v_1,v_2>0$ or if one of the two parameters is set equal to zero. 

The latter means that  we need to study the case $h_1\geq h_6$ and $h_2\leq h_7$. We now perform the substitution
 $h_1=h_6+v_1$ and $h_7=h_2+v_2$, again with one of $v_1$ or $v_2$ nonzero. 
We observe that if $h_5\geq h_3$, then again $\det H_3$ is positive. So we restrict to $h_3>h_5$ and perform the substitution $h_3=h_5+v_3$ into $\det H_3$. 

When we do that, $\det H_3$ has coefficients of both signs, and therefore the sign is not clear. We have still to impose $c_0>0$ for this scenario. 
We perform the substitutions into $c_0$ and obtain:
\begin{multline*}
c_0 =\left( -2\,h_{2}h_{4}h_{5}-2\,h_{2}h_{4}v_{3}-h_{2}h_5^{2}-h_{5}h_{2}v_{3}-h_4h_5^{2}-h_{4}h_{5}v_{3}-h_{4}v_{2}v_{3} \right) v_{1} \\+h_{4}h_{5}^{2}v_{2}
+2\,h_{4}h_{5}h_{6}v_{2}+h_{4}h_{5}v_{2}v_{3}+h_{5}^{2}h_{6}v_{2}+h_{5}h_{6}v_{2}v_{3}
\end{multline*}
For $c_0$ to be positive, we need $v_1$ to be smaller than the root of $c_0$ seen as a polynomial in $v_1$, which is:
$$z_1 := \frac{h_{5}v_{2} \left( h_{4}h_{5}+2\,h_{4}h_{6}+h_{4}v_{3}+h_{6}h_{5}+h_{6}v_{3} \right) }{2\,h_{2}h_{4}h_{5}+2\,h_{2}h_{4}v_{3}+h_{2}h_5^{2}+h_{2}h_{5}v_{3}+h_{4}h_{5}^{2}+h_{4}h_{5}v_{3}+h_{4}v_{2}v_{3}}. $$
Now, we need to check whether $\det H_3$ can be negative when $v_1$ is smaller than $z_1$. To check that, we make the substitution 
$$ v_1 = \frac{\mu}{\mu+1} z_1.$$
Then any number in the interval $[0,z_1)$ is of this form for some nonnegative $\mu$.
We perform this substitution in Maple, gather the numerator of the resulting $\det H_3$, and confirm that all signs are positive. Further $\det H_3$ is positive even if some of $\mu,v_2,v_3$ are zero.

The other case $h_2\geq h_7$ and $h_1\leq h_6$ is analogous by the symmetry of the system. 
This finishes the argument, since we have explored all possibilities for $c_0>0$, and they all give that $\det H_3$ is positive. 
\end{proof}

\subsection{Network \ref{eq:N2}}
\label{sec:network-N2}

We use the following ordering: $x_1,x_2,x_3,x_4,x_5,x_6,x_7$ for the concentration of $K,F,S_0,S_1,S_2,K S_0,F S_1$ respectively. Under
this ordering the stoichiometric matrix, the kinetic order matrix and
the matrix $E_2$ are
\begin{displaymath}
  N_2 = \left[
    \begin{array}{rrrrrr}
      -1 & 1 & 0 & 0 & 0 & 0 \\
      0 & 0 & 0 & 0 & -1 & 1 \\
      -1 & 0 & 0 & 0 & 0 & 1 \\
      0 & 1 & -1 & 1 & -1 & 0 \\
      0 & 0 & 1 & -1 & 0 & 0 \\
      1 & -1 & 0 & 0 & 0 & 0 \\
      0 & 0 & 0 & 0 & 1 & -1 \\
    \end{array}
  \right], 
  Y_2 = \left[
    \begin{array}{cccccc}
      1 & 0 & 1 & 0 & 0 & 0 \\
      0 & 0 & 0 & 1 & 1 & 0 \\
      1 & 0 & 0 & 0 & 0 & 0 \\
      0 & 0 & 1 & 0 & 1 & 0 \\
      0 & 0 & 0 & 1 & 0 & 0 \\
      0 & 1 & 0 & 0 & 0 & 0 \\
      0 & 0 & 0 & 0 & 0 & 1 \\
    \end{array}
  \right] \text{ and }
  E_2 = \left[
    \begin{array}{cc}
      1 & 0 \\
      1 & 0 \\
      0 & 1 \\
      0 & 1 \\
      1 & 0 \\
      1 & 0 \\
    \end{array}
  \right].
\end{displaymath}
With this parametrization, the Jacobian of the system evaluated at a
steady state defined by $(h_1,\dots,h_7,\lambda)$ is: 
\begin{displaymath}
  J_2(h,\lambda) = \left[
    \begin{array}{ccccccc}
      -h_{1} & 0 & -h_{3} & 0 & 0 & h_{6} & 0 \\
      0 & -h_{2} & 0 & -h_{4} & 0 & 0 & h_{7} \\
      -h_{1} & 0 & -h_{3} & 0 & 0 & 0 & h_{7} \\
      -h_{1} \lambda & h_{2} (-1+\lambda) & 0 & h_{4} (-1-\lambda) & h_{5} \lambda & h_{6} & 0 \\
      h_{1} \lambda & -h_{2} \lambda & 0 & h_{4} \lambda & -h_{5} \lambda & 0 & 0 \\
      h_{1} & 0 & h_{3} & 0 & 0 & -h_{6} & 0 \\
      0 & h_{2} & 0 & h_{4} & 0 & 0 & -h_{7} \\
    \end{array}
  \right].
\end{displaymath}
 $\det H_3(h,\lambda)$ contains only positive
  monomials and in particular it does not vanish for any positive $h,\lambda$.

  \subsection{Network \ref{eq:N3}}
\label{sec:network-N3}
We use the following ordering: $x_1,x_2,x_3,x_4,x_5,x_6,x_7$ for the concentration of $K,F,S_0,S_1,S_2,K S_0,K S_1$ respectively. 
Under
this ordering the stoichiometric matrix, the kinetic order matrix and
the matrix $E_3$ are 
\begin{displaymath}
  N_3 = \left[
    \begin{array}{rrrrrr}
      -1 & 1 & -1 & 1 & 0 & 0 \\
      0 & 0 & 0 & 0 & 0 & 0 \\
      -1 & 0 & 0 & 0 & 0 & 1 \\
      0 & 1 & -1 & 0 & 1 & -1 \\
      0 & 0 & 0 & 1 & -1 & 0 \\
      1 & -1 & 0 & 0 & 0 & 0 \\
      0 & 0 & 1 & -1 & 0 & 0 \\
    \end{array}
  \right], 
  Y_3 = \left[
    \begin{array}{cccccc}
      1 & 0 & 1 & 0 & 0 & 0 \\
      0 & 0 & 0 & 0 & 1 & 1 \\
      1 & 0 & 0 & 0 & 0 & 0 \\
      0 & 0 & 1 & 0 & 0 & 1 \\
      0 & 0 & 0 & 0 & 1 & 0 \\
      0 & 1 & 0 & 0 & 0 & 0 \\
      0 & 0 & 0 & 1 & 0 & 0 \\
    \end{array}
  \right] \text{ and }
  E_3 = \left[
    \begin{array}{cc}
      1 & 0 \\
      1 & 0 \\
      0 & 1 \\
      0 & 1 \\
      0 & 1 \\
      1 & 0 \\
    \end{array}
  \right].
\end{displaymath}
With this parametrization, the Jacobian of the system evaluated at a
steady state defined by $(h_1,\dots,h_7,\lambda)$ is: 
\begin{displaymath}
  J_3(h,\lambda) = \left[
    \begin{array}{ccccccc}
      h_{1} (-1-\lambda) & 0 & -h_{3} & -h_{4} \lambda & 0 & h_{6} & h_{7} \lambda \\
      0 & 0 & 0 & 0 & 0 & 0 & 0 \\
      -h_{1} & h_{2} & -h_{3} & h_{4} & 0 & 0 & 0 \\
      -h_{1} \lambda & h_{2} (-1+\lambda) & 0 & h_{4} (-1-\lambda) & h_{5} \lambda & h_{6} & 0 \\
      0 & -h_{2} \lambda & 0 & 0 & -h_{5} \lambda & 0 & h_{7} \lambda \\
      h_{1} & 0 & h_{3} & 0 & 0 & -h_{6} & 0 \\
      h_{1} \lambda & 0 & 0 & h_{4} \lambda & 0 & 0 & -h_{7} \lambda \\
    \end{array}
  \right].
\end{displaymath}
 $\det H_3(h,\lambda)$ contains only positive
  monomials and in particular it does not vanish for any positive $h,\lambda$.

\subsection{Network \ref{eq:N4}}
\label{sec:network-N4}
We use the following ordering: $x_1,x_2,x_3,x_4,x_5,x_6,x_7$ for the concentration of $K,F,S_0,S_1,S_2,K S_1,F S_1$ respectively. 
Under this ordering the stoichiometric matrix, the kinetic order
matrix and the matrix $E_4$ are 
\begin{displaymath}
  N_4 = \left[
    \begin{array}{rrrrrrr}
      0 & -1 & 1 & 0 & 0 & 0 \\
      0 & 0 & 0 & 0 & -1 & 1 \\
      -1 & 0 & 0 & 0 & 0 & 1 \\
      1 & -1 & 0 & 1 & -1 & 0 \\
      0 & 0 & 1 & -1 & 0 & 0 \\
      0 & 1 & -1 & 0 & 0 & 0 \\
      0 & 0 & 0 & 0 & 1 & -1 \\
    \end{array}
  \right], 
  Y_4 = \left[
    \begin{array}{cccccc}
      1 & 1 & 0 & 0 & 0 & 0 \\
      0 & 0 & 0 & 1 & 1 & 0 \\
      1 & 0 & 0 & 0 & 0 & 0 \\
      0 & 1 & 0 & 0 & 1 & 0 \\
      0 & 0 & 0 & 1 & 0 & 0 \\
      0 & 0 & 1 & 0 & 0 & 0 \\
      0 & 0 & 0 & 0 & 0 & 1 \\
    \end{array}
  \right] \text{ and }
  E_4 = \left[
    \begin{array}{cc}
      1 & 0 \\
      0 & 1 \\
      0 & 1 \\
      0 & 1 \\
      1 & 0 \\
      1 & 0 \\
    \end{array}
  \right].
\end{displaymath}
With this parametrization, the Jacobian of the system evaluated at a
steady state defined by $(h_1,\dots,h_7,\lambda)$ is: 
\begin{displaymath}
  J_4(h,\lambda) = \left[
    \begin{array}{ccccccc}
      -h_{1} \lambda & 0 & 0 & -h_{4} \lambda & 0 & h_{6} \lambda & 0 \\
      0 & -h_{2} & 0 & -h_{4} & 0 & 0 & h_{7} \\
      -h_{1} & 0 & -h_{3} & 0 & 0 & 0 & h_{7} \\
      h_{1} (1-\lambda) & h_{2} (-1+\lambda) & h_{3} & h_{4} (-1-\lambda) & h_{5} \lambda & 0 & 0 \\
      0 & -h_{2} \lambda & 0 & 0 & -h_{5} \lambda & h_{6} \lambda & 0 \\
      h_{1} \lambda & 0 & 0 & h_{4} \lambda & 0 & -h_{6} \lambda & 0 \\
      0 & h_{2} & 0 & h_{4} & 0 & 0 & -h_{7} \\
    \end{array}
  \right].
\end{displaymath}
 $\det H_3(h,\lambda)$ contains only positive
  monomials and in particular it does not vanish for any positive $h,\lambda$.

\section{Discussion and Outlook: Does the full network admit Hopf bifurcations?}
\label{sec:discussion}

  As discussed in the Introduction, in \cite[Section~5.36]{osc-010} a
  simplification of the mass action model of Fig.~\ref{fig:cycles-dd} is
  examined and the authors provide steady state concentration values
  and rate constants of a candidate Hopf bifurcation point. Here we
  want to explain how this point fails to give a pair of purely imaginary eigenvalues (Proposition
  \ref{prop:yang}). 
 
  In \cite{osc-010}, the authors consider the following irreversible
  version of the network~\ref{eq:network}:
  \begin{equation}
    \tag{$\mathcal{N}_{f}$}
    \label{eq:network_ir}
    \begin{split}
      S_0 + K \ce{->[\k_1]} KS_0 \ce{->[\k_3]} S_1+K
      \ce{->[\k_4]} KS_1 \ce{->[\k_6] } S_2+K \\
      S_2 + F  \ce{->[\k_7]} FS_2 \ce{->[\k_9]} S_1+F
      \ce{->[\k_{10}]} FS_1 \ce{->[\k_{12}]} S_0+F.
    \end{split}
  \end{equation}
  In network~\ref{eq:network_ir} the constants $\k_2$, $\k_5$, $\k_8$
  and $\k_{11}$ describing the `backward' reactions in the
  full network~\ref{eq:network} have been assigned the value
  zero. Hence the matrix $E_f$ whose columns generate the
  cone~(\ref{eq:lin-problem}) is
  \begin{displaymath}
    E_f^T=
    \left[
      \begin{array}{cccccccccccc}
        1 & 0 & 1 & 0 & 0 & 0 & 0 & 0 & 0 & 1 & 0 & 1 \\
        0 & 0 & 0 & 1 & 1 & 1 & 1 & 1 & 1 & 0 & 0 & 0
      \end{array}
    \right].
  \end{displaymath}
  Using  $x_1,x_2,x_3,x_4,x_5,x_6,x_7,x_8,x_9$ for the concentrations of $S_0,K, KS_0,S_1, KS_1, S_2,F, FS_2, FS_1$ respectively one obtains the Jacobian matrix
  \begin{displaymath}
    \tiny
    J_f(h,\lambda)= 
    \left[
      \begin{array}{ccccccccc}
        -h_{1} \lambda_{1} & -h_{2} \lambda_{1} & 0 & 0 & 0 & 0 & 0 & 0 & h_{9} \lambda_{1} \\
        -h_{1} \lambda_{1} & h_{2} (-\lambda_{1}-\lambda_{2}) & h_{3} \lambda_{1} & -h_{4} \lambda_{2} & h_{5} \lambda_{2} & 0 & 0 &0 & 0 \\
        h_{1} \lambda_{1} & h_{2} \lambda_{1} & -h_{3} \lambda_{1} & 0 & 0 & 0 & 0 & 0 & 0 \\
        0 & -h_{2} \lambda_{2} & h_{3} \lambda_{1} & h_{4} (-\lambda_{1}-\lambda_{2}) & 0 & 0 & -h_{7} \lambda_{1} & h_{8} \lambda_{2}& 0 \\
        0 & h_{2} \lambda_{2} & 0 & h_{4} \lambda_{2} & -h_{5} \lambda_{2} & 0 & 0 & 0 & 0 \\
        0 & 0 & 0 & 0 & h_{5} \lambda_{2} & -h_{6} \lambda_{2} & -h_{7} \lambda_{2} & 0 & 0 \\
        0 & 0 & 0 & -h_{4} \lambda_{1} & 0 & -h_{6} \lambda_{2} & h_{7} (-\lambda_{1}-\lambda_{2}) & h_{8} \lambda_{2} & h_{9} \lambda_{1}\\
        0 & 0 & 0 & 0 & 0 & h_{6} \lambda_{2} & h_{7} \lambda_{2} & -h_{8} \lambda_{2} & 0 \\
        0 & 0 & 0 & h_{4} \lambda_{1} & 0 & 0 & h_{7} \lambda_{1} & 0 & -h_{9} \lambda_{1} \\
      \end{array}
    \right].
  \end{displaymath}
  The authors of \cite{osc-010} provide the candidate point where
  $\det H_{5}(h^*,\lambda^*)=0$:
  \begin{align*}
    h_{1}^* &=  9.15394021721585 \cdot 10^{-6}  &
    h_{2}^* &=  8.438690345203897 \cdot 10^{-6} &
    h_{3}^* &=  9.15394021721585 \cdot 10^{-6}  \\
    h_{4}^* &=  9.15394021721585 \cdot 10^{-6}  &
    h_{5}^* &=  0.0000234589 &
    h_{6}^* &=  0.00391442 \\
    h_{7}^* &=  0.0625077 &
    h_{8}^* &=  0.00391442 &
    h_{9}^* &=  1 \\
    \lambda_{1}^* &=  1 &
    \lambda_{2}^* &=  1.
  \end{align*}
We find that
  \begin{align*}
    \det H_1(h^*,\lambda^*) &=  \phantom{-}1.13292 &
    \det H_2(h^*,\lambda^*) &=  0.0803339 \\
    \det H_3(h^*,\lambda^*) &=  \phantom{-}1.7440236556291417 \cdot 10^{-6} &
    \det H_4(h^*,\lambda^*) &=  2.5421805536611004 \cdot 10^{-15} \\
    \det H_5(h^*,\lambda^*) &= -1.3900880766102185 \cdot 10^{-42} \\
    a_6(h^*,\lambda^*) &= -1.682281311658486 \cdot 10^{-18}.
  \end{align*}
  While $\det H_4(h^*,\lambda^*)$, $\det H_5(h^*,\lambda^*)$ and
  $a_6(h^*,\lambda^*)$ are all very small numbers, $\det
  H_5(h^*,\lambda^*)$ is by orders of magnitude smaller. And if
  $10^{-42} \approx 0$, then $\det H_4(h^*,\lambda^*)$ and
  $a_6(h^*,\lambda^*)$ are  nonzero. 
 In particular $a_6(h^*,\lambda^*)<0$ and hence this point gives rise
 to a pair of symmetric real eigenvalues.
 We have additionally verified our claim by considering
 the set of corresponding eigenvalues (computed with Mathematica):
 \begin{displaymath}
   \begin{split}
     \big\{ 
     &-1.06643, -0.0661828, -0.000208571, -0.0000990197, 0.0000339721,
     -0.0000339721, \\
     &-4.701752723492088\cdot 10^{-16}, 1.2668321503904372\cdot 10^{-17},
     1.495648385287835\cdot 10^{-18}
     \big\}.
   \end{split}
 \end{displaymath}
 All eigenvalues are real numbers, three eigenvalues are $\approx 0$
 (as expected as the system has three conservation relations) and the
 real values $0.0000339721$ and $-0.0000339721$ give a pair of symmetric eigenvalues (up
 to $10$ digits).
 
 Hence, this justifies our claim that the question of
 whether or not a Hopf bifurcation can occur in the full
 network~\eqref{eq:network} is still open. For the full network,
 $\det H_1,\dots,\det H_{s-2}$ are all positive, and hence the
 problem is reduced to making $\det H_{s-1}$ vanish while having
 $a_s$ positive. Network \eqref{eq:N1} is the smallest simplified
 network where both $\det H_{s-1}$ and $a_s$ attain both signs and
 hence poses the same challenge.
 
 For networks admitting pairs of purely imaginary eigenvalues, a
 typical approach in this situation is to find values for which
 $\det H_{s-1}$ vanishes, and then verify that $a_s$ is positive
 when evaluated at  these values. But if the two conditions happen
 to be incompatible, as it might be the case for the full
 network~\eqref{eq:network}, then it is unclear how to proceed.
 Automated approaches \cite{osc-010} cannot handle the analysis of
 the signs of these two large polynomials at the same time. Here,
 by combining guided heavy symbolic computations with manual
 intervention based on the visual inspection of $a_s$, we have
 managed to prove that the two sign conditions are incompatible for
 network \eqref{eq:N1}.

 The idea is to reduce the problem into checking the sign of the
 coefficients of a polynomial. We first break the condition $a_s>0$
 into subcases involving simple inequalities between the
 parameters. Then, one uses the inequalities to substitute one
 parameter in $H_{s-1}$ by a new parameter that is allowed to take
 any positive value (e.g. $k_1>k_2$ leads to  the substitution
 $k_1=k_2+a$, such that the new constraints are
 $k_2,a>0$). Finally, $H_{s-1}$ is transformed into a polynomial
 (or rational function) that only needs to be evaluated at positive
 values to guarantee $a_s>0$. If all coefficients of $H_{s-1}$ are
 positive, then trivially the polynomial is positive. If some
 coefficients are negative, then further exploration is required.
    
 This strategy has worked nicely for network \eqref{eq:N1}. However,
 substitutions of the type $k_1=k_2+a$ typically increase
 dramatically the number of terms as a polynomial of the form
 $k_1^d\, p(\k,x)$ becomes $(k_2+a)^d\, p(\k,x)$. We also performed
 substitutions of the form $k_1 = \tfrac{\mu}{\mu+1}
 \tfrac{p(\k,x)}{q(\k,x)}$ with $\mu>0$, which introduced large
 denominators.

 All this illustrates the difficulty in the analysis of the full
 network: the computation of the Hurwitz determinants is very
 demanding, but a posterior analysis, in par with the analysis of
 network \eqref{eq:N1}, is prohibitive. Nevertheless, we think ideas
 introduced here in the analysis of network \eqref{eq:N1} might be
 applied to other networks posing similar challenges, and maybe
 similar ideas end up being successful for the full network either
 after increasing computational power or by depicting new strategies
 to simplify computations.

 As explained in Remark~\ref{rk:simplify}, the existence of periodic
 orbits in simplified models can be lifted to more complex models
 under certain network modifications
 \cite{Banaji18:oscillations}. However, the non-existence of periodic
 orbits or Hopf bifurcations does not, in principle, tell us anything
 about complex models including the simplified models in some
 form. Considering the demanding computational cost involved in
 describing the dynamics of a reaction network for arbitrary
 parameter values, it would be of great help to have a better
 understanding of what type of network operations or parameter
 regimes guarantee that the non-existence of a behavior in a
 simplified model is preserved in the more complex model. For
 example, we are not aware of any result nor reasoning we could
 employ to ensure that our conclusions on networks
 \eqref{eq:N1}-\eqref{eq:N4} are preserved after introducing
 unbinding reactions with a small rate constant.

\medskip
\subsection*{Supplementary Material}
The Maple computations for the proof of Theorem~\ref{thm:main} are provided in the file {\bf `SupplMat1.mw'}. For the convenience of the readers without access to Maple, we provide as well as pdf version of the file.

\bigskip
\subsection*{Acknowledgements} CC acknowledges funding from Deutsche
Forschungsgemeinschaft, 284057449. EF   acknowledges funding from the
Danish Research Council for Independent Research.

The authors thank Ang{\'e}lica Torres for providing source code to
compute Hurwitz determinants
and Alan Rendall for helpful discussions about (simple) Hopf
bifurcations.
The authors thank the hospitality Erwin Schr{\"o}dinger International
Institute for Mathematics and Physics in Vienna, where this project
was started during the workshop `Advances in Chemical Reaction Network
Theory'.

\small 
\bibliographystyle{plain}
\bibliography{crnt}
 \end{document}